\let\oldvec\vec
\documentclass{llncs}
\let\vec\oldvec

\usepackage[ruled,vlined,linesnumbered,commentsnumbered]{algorithm2e}
\usepackage{algcompatible}
\usepackage{amsmath,amssymb,amsfonts}

\usepackage{graphicx}

\newcommand {\ignore} [1] {}

\newcommand{\sem}    {\setminus}
\newcommand{\subs}   {\subseteq}
\newcommand{\empt}  {\emptyset}

\newcommand{\f}   {\frac}
\newcommand{\p}   {\partial}

\newcommand{\opt}   {\sf opt}

\newcommand{\A}    {\mathbb{A}}
\newcommand{\B}    {\mathbb{B}}
\newcommand{\C}    {\mathbb{C}}

\newcommand{\CC}    {{\cal C}}
\newcommand{\DD}    {{\cal D}}
\newcommand{\FF}    {{\cal F}}

\newcommand{\be}   {\beta}

\newcommand{\eps} {\epsilon}

\newcommand{\Ga} {\Gamma}
\newcommand{\De} {\Delta}

\newcommand{\SC} {{\sc Subset $k$-Connectivity}}
\newcommand{\RSC} {{\sc Rooted Subset $k$-Connectivity}}

\newcommand{\SeC} {{\sc Set-Cover}}

\newcommand{\RSCA} {{\sc Rooted Subset Connectivity Augmentation}}

\begin{document}

\title{Approximating $k$-connected $m$-dominating sets}

\author{Zeev Nutov}
\institute{The Open University of Israel, {\tt nutov@openu.ac.il}}

\maketitle

\begin{abstract}
A subset $S$ of nodes in a graph $G$ is a {\bf $k$-connected $m$-dominating set} ({\bf $(k,m)$-cds})
if the subgraph $G[S]$ induced by $S$ is $k$-connected and every $v \in V \setminus S$ has at least $m$ neighbors in $S$.
In the {\sc $k$-Connected $m$-Dominating Set} ({\sc $(k,m)$-CDS}) 
problem the goal is to find a minimum weight $(k,m)$-cds in a node-weighted graph.
For $m \geq k$ we obtain the following approximation ratios. 
For general graphs our ratio $O(k \ln n)$ improves  the previous best ratio $O(k^2 \ln n)$ of \cite{N-CDS}
and matches the best known ratio for unit weights of \cite{ZTHD}.
For unit disc graphs we improve the ratio $O(k \ln k)$ of \cite{N-CDS} to 
$\min\left\{\f{m}{m-k},k^{2/3}\right\} \cdot O(\ln^2 k)$  -- this is the first sublinear ratio for the problem, 
and the first polylogarithmic ratio $O(\ln^2 k)/\eps$ when $m \geq (1+\eps)k$;
furthermore, we obtain ratio $\min\left\{\f{m}{m-k},\sqrt{k}\right\} \cdot O(\ln^2 k)$ for uniform weights.
These results are obtained by showing the same ratios for the {\SC} problem when the 
set $T$ of terminals is an $m$-dominating set with $m \geq k$.
\end{abstract}


\section{Introduction} \label{s:intro}

All graphs in this paper are assumed to be simple, unless stated otherwise.
A (simple) graph is {\bf $k$-connected} if it has 
$k$ pairwise internally node disjoint paths between every pair of its nodes;
in this case the graph has at least $k+1$ nodes. 
A subset $S$ of nodes in a graph $G$ is a {\bf $k$-connected set} if the subgraph $G[S]$ induced by $S$ is $k$-connected;
$S$ is an {\bf $m$-dominating set} if every $v \in V \setminus S$ has at least $m$ neighbors in $S$.
If $S$ is both $k$-connected and $m$-dominating set then 
$S$ is a {\bf $k$-connected $m$-dominating set}, or {\bf $(k,m)$-cds} for short.
A graph is a {\bf unit-disk graph}  if its nodes can be located in 
the Euclidean plane such that there is an edge between nodes $u$ and $v$ 
iff the Euclidean distance between $u$ and $v$ is at most $1$.
We consider the following problem for $m \geq k$ both in general graphs and in unit-disc graphs.


\begin{center} \fbox{\begin{minipage}{0.965\textwidth} \noindent
{\sc $k$-Connected $m$-Dominating Set} ({\sc $(k,m)$-CDS}) \\
{\em Input:}  A graph $G=(V,E)$ with node weights $\{w_v:v \in V\}$ and integers $k,m$. \\
{\em Output:}   A minimum weight $(k,m)$-cds $S \subseteq V$.
\end{minipage}}\end{center}

For motivation  we refer the reader to recent papers of 
Zhang, Zhou, Mo, and Du \cite{ZW} and of Fukunaga \cite{Fu}, where they obtained in unit-disc graphs
ratios $O(k^3 \ln k)$ and  $O(k^2 \ln k)$, respectively.
This was improved to $O(k \ln k)$ in \cite{N-CDS}, where is also given ratio $O(k^2 \ln n)$ in general graphs.
Our main results is:

\begin{theorem} \label{t:CDS}
{\sc $(k,m)$-CDS} with $m \geq k$ admits the following approximation ratios: 
$O(k \ln n)$ in general graphs, $\min\left\{\f{m}{m-k},k^{2/3}\right\} \cdot O(\ln^2 k)$ in unit disc graphs,
and $\min\left\{\f{m}{m-k},\sqrt{k}\right\} \cdot O(\ln^2 k)$ in unit disc graphs with unit weights.
\end{theorem}

For general graphs our ratio $O(k \ln n)$ improves  the previous ratio $O(k^2 \ln n)$ of \cite{N-CDS}
and matches the best known ratio for unit weights of \cite{ZTHD}.
For unit disc graphs our ratio $\min\left\{\f{k}{m-k},k^{2/3}\right\} \cdot O(\ln^2 k)$
improves the previous best ratio $O(k \ln k)$ of \cite{N-CDS}; this is the first sublinear ratio for the problem, 
and for any constant $\eps>0$ and $m=k(1+\eps)$ the first polylogarithmic ratio $O(\ln^2 k)/\eps$.

Let us say that a graph with a set $T$ of terminals and a root $r \in T$ is 
{\bf $k$-$(T,r)$-connected} if it has $k$ internally node disjoint $rt$-paths for every $t \in T \sem \{r\}$.
Similarly, a graph is {\bf $k$-$T$-connected} if it has $k$ internally node disjoint $st$-paths for every $s,t \in T$.
A reason why the case $m \geq k$ is easier than the case $m<k$ 
is given in the following statement (a proof can be found in \cite{ZW,Fu,N-CDS}).

\begin{lemma} \label{l:kc}
Let $T$ be a $k$-dominating set in a graph $H=(U,F)$.
If $H$ is $k$-$(T,r)$-connected then $H$ is $k$-$(U,r)$-connected;
if $H$ is $k$-$T$-connected then $H$ is $k$-connected.
\end{lemma}

The above lemma implies that in the case $m \geq k$ {\sc $(k,m)$-CDS} partial solutions have the property 
that the union of a partial solution and a feasible solution is always feasible -- this enables to construct the solution iteratively.
Specifically, most algorithms for the case $m \geq k$ start by computing just an $m$-dominating set $T$;
the best ratios for {\sc $m$-Dominating Set} are $\ln(\De+m)$ in general graphs \cite{Fo} and $O(1)$ in unit disc graphs \cite{Fu}.
By invoking just these ratios, Lemma~\ref{l:kc} enables to reduce {\sc $(k,m)$-CDS} with $m \geq k$ 
to following (node weighted) problem:

\begin{center} \fbox{\begin{minipage}{0.965\textwidth} \noindent
{\SC}  \\
{\em Input:}  A graph $G=(V,E)$ with node-weights $\{w_v:v \in V\}$, a set $T \subseteq V$ of terminals, and an integer $k$. \\
{\em Output:}   A minimum weight $k$-$T$-connected subgraph of $G$. 
\end{minipage}}\end{center}

The ratios for this problem are usually expressed in terms of the best known ratio $\be$ for the following problem
(in both problems we will assume w.l.o.g. that $w_v=0$ for all $v \in T$):

\begin{center} \fbox{\begin{minipage}{0.965\textwidth} \noindent
{\RSC}  \\
{\em Input:}  A graph $G=(V,E)$ with node-weights $\{w_v:v \in V\}$, a set $T \subseteq V$ of terminals, 
a root node $r \in T$, and an integer $k$. \\
{\em Output:}   A minimum weight $k$-$(T,r)$-connected subgraph of $G$. 
\end{minipage}}\end{center}

Currently, $\be=O(k^2\ln |T|)$ \cite{N-TALG}.
From previous work it can be deduced that {\SC} with $|T| \geq k$ admits ratio $\be+k^2$.
Add a new root node $r$ connected to a set $R \subs T$ of $k$ nodes by edges of cost zero.
Then compute a $\beta$-approximate solution to the obtained {\RSC} instance. Finally, augment this solution 
by computing for every $u,v \in R$ a min-weight set of $k$ internally disjoint $uv$-paths.
For the {\sc $(k,m)$-CDS} problem with $m \geq k$ this already gives ratio $\be+k^2=O(k^2\ln |T|)$ in general graphs.
For the special case when $T$ is a $k$-dominating set the ratio $\be+k^2$ was improved in \cite{N-CDS} to $\be+k-1$,
since then in the final step it is sufficient to compute a min-weight set of $k$ internally disjoint $uv$-paths
for pairs that form a forest on $R$.

We now consider unit disc graphs. Zhang et. al. \cite{ZW} showed that 
any $k$-connected unit disc graph has a $k$-connected spanning subgraph of maximum degree $\leq 5k$.
This implies that the node weighted case is reduced with a loss of factor $O(k)$ 
to the case of node induced edge costs -- when $c_{uv}=w_u+w_v$ for every edge $e=uv \in E$.
The edge costs version of {\SC} admits ratio $O(k^2\ln k)$, which gives ratio $O(k^3 \ln k)$ 
for {\sc $(k,m)$-CDS} with $m \geq k$ in unit disc graphs.
Fukunaga \cite{Fu} obtained ratio $O(k^2\ln k)$ using a different approach -- he considered the {\RSCA} problem,
when $G[T]$ is $\ell$-$(T,r)$-connected and we seek a minimum weight $S \subs V \sem T$ such that $G[T \cup S]$ is 
$(\ell+1)$-$(T,r)$-connected. In \cite{N-TALG} it is shown that the augmentation problem decomposes into $O(k)$ ``uncrossable'' subproblems,
and Fukunaga \cite{Fu} designed a primal-dual $O(1)$-approximation algorithm for such an uncrossable subproblem in unit disc graphs.
This gives ratio $O(\ell)$ for {\RSCA} in unit disc graphs.
Furthermore, using the so called ``backward augmentation analysis'' Fukunaga showed that since his approximation is w.r.t. an LP, 
then sequentially increasing the $T$-connectivity by $1$ invokes only a factor of $O(\ln k)$, thus obtaining ratio
$O(k \ln k)$ for {\RSCA}. He then combined this result with a decomposition of the {\SC} problem into $k$ {\RSC} problems, 
and obtained ratio $O(k^2 \ln k)$.
As was mentioned, in \cite{N-CDS} it is proved that ratio $\be$ for {\RSC} implies ratio $\be+k-1$ for {\sc $(k,m)$-CDS} with $m \geq k$,
which improves the ratio to $O(k \ln k)$.

However, it seems that previous reductions and methods alone do not enable to obtain 
ratio better than $O(k^2 \ln |T|)$ in general graphs, or a sublinear ratio in unit disc graphs.
These algorithm rely on the ratios and decompositions for the {\sc Rooted}/{\SC} problems from \cite{N-TALG,N-subs},
but these do not consider the specific feature relevant to {\sc $(k,m)$-CDS}  -- 
that the set $T$ of terminals is an $m$-dominating set;
note that then {\SC} is equivalent to the problem of finding the 
lightest $k$-connected subgraph containing $T$, by Lemma~\ref{l:kc}.
Here we change this situation by asking the following question: \\
{\em If the set $T$ of terminals is an $m$-dominating set with $m \geq k$,
what approximation ratios can we achieve for (node weighted) {\SC}?} \\
Our answer to this question is given in the following theorem, which is of independent interest,
and note that it implies Theorem~\ref{t:CDS}.

\begin{theorem} \label{t:SC}
The (node weighted) {\SC} problem such that $T$ is an $m$-dominating set with $m \geq k$ admits the following approximation ratios: 
$O(k \ln n)$ in general graphs, $O(\ln^2 k) \cdot \min\left\{\f{m}{m-k+1},k^{2/3}\right\}$ in unit disc graphs,
and $O(\ln^2 k) \cdot \min\left\{\f{m}{m-k+1},\sqrt{k}\right\}$ in unit disc graphs with unit weights.
\end{theorem}

In the proof of Theorem~\ref{t:SC} we use several results and ideas from previous works \cite{N-TALG,N-subs,ZW,Fu,N-CDS}. 
As was mentioned, the best ratios for the {\SC} are derived via reductions of \cite{N-subs,N-CDS} from the ratios for the {\RSC} problem, 
so we will consider the latter problem; the currently best known ratio for this problem is $O(k^2 \ln |T|)$ \cite{N-TALG}.
The algorithm of \cite{N-TALG} works in $\ell$ iterations, where at iteration $\ell=0,\ldots,k-1$ it considers the augmentation problem of increasing 
the connectivity from $\ell$ to $\ell+1$. This is equivalent to covering a certain family $\FF$ of ``tight sets'' (a.k.a. ``deficient sets''), 
and  the algorithm of \cite{N-TALG} decomposes this problem into $O(\ell)$ uncrossable family covering problems;
the ratio for covering each uncrossable family is $O(\ln n)$ in general graphs \cite{N-TALG} and $O(1)$ in unit disc graphs \cite{Fu}.

However, a more careful analysis of the \cite{N-TALG} algorithm reveals that in fact the  
number of uncrossable families is $O(\ell/q)+1$, where $q$ is the minimum number of terminals in a tight set.
Specifically, the algorithm has an ``inflation phase'' that works towards reaching $q \geq \ell+1$ -- 
in which case the entire family of tight sets becomes uncrossable, by repeatedly covering $O(\ell/q)$ uncrossable families to double $q$. 
Hence if $q_0$ is the initial value of $q$, the total number of uncrossable families that the algorithm covers is $1$ plus order of
$\f{\ell}{q_0}\left(1+\f{1}{2}+\f{1}{4}+ \cdots\right)=O(\ell/q_0)$.
Note that a large part of the uncrossable families are covered at the beginning -- when $q$ is small.
One of our contributions is designing a different ``lighter'' inflation algorithms for increasing the parameter $q$.
These algorithms just aim to cover the inclusion minimal tight sets by adding a light set $S$ of nodes, 
and then add $S$ to the set $T$ of terminals;
if $T$ is a $k$-dominating set then adding any set $S$ to $T$ does not make the problem harder, by Lemma~\ref{l:kc}.

Our algorithms for covering inclusion minimal tight sets reduce the problem to a set covering type problem.
In the case of general graphs the reduction is to a special case considered in \cite{N-SL} of the {\sc Submodular Covering} problem;
the ratio invoked by this procedure is only $O(\ln n)$ and if we apply it $p=\max\{2k-m-1,1\}$ times 
then we get $q \geq m-\ell+p(k-\ell) \geq k$ for all $\ell=0,\ldots,k-1$.
In fact, we apply this procedure before considering the augmentation problems, 
but it guarantees that $q \geq k$ through all augmentation iterations.
The same procedure applies in the case of unit disc graphs, 
but to avoid the dependence on $n$ in the ratio we use a different procedure. 
Specifically, we use the result of Zhang et. al. \cite{ZW} 
that minimally $k$-connected unit disc graph has maximum degree $\leq 5k$,
to reduce the problem of covering the family of tight sets to the {\SeC} problem with soft capacities.
This approach gives ratio $\min\left\{\f{m}{m-k},k^{2/3}\right\} \cdot O(\ln^2 k)$.

In the rest of the paper we prove Theorem~\ref{t:SC};
Section \ref{s:gen} considers general graphs and Section \ref{s:udg} considers unit disc graphs.

\section{General graphs} \label{s:gen}

While edge-cuts of a graph correspond to node subsets, a natural way to
represent a node-cut of a graph is by a pair of sets, as follows.

\begin{definition}
An ordered pair $\A=(A,A^+)$ of subsets of $V$ with $A \subs A^+$ is called a {\bf biset};
the set $\p\A=A^+ \sem A$ is called the {\bf cut} of $\A$. 
We say that $\A$ is a $(T,r)$-biset if $A \cap T \neq \empt$ and $r \in V \sem A^+$.
An {\bf edge covers a biset $\A$} if it has one end in $A$ and the other in $V \sem A^+$.
For an edge set/graph $J$ let $d_J(\A)$ denote the number of edges in $J$ that cover $\A$.
\end{definition}

By  Menger’s Theorem $G$ 
is $k$-$(T,r)$-connected iff $|\p\A|+ d_{G[T]}(\A) \geq k$ holds for every $(T,r)$-biset $\A$.
Given a {\RSC} in\-stance, we say that a $(T,r)$-biset $\A$ is a {\bf deficient biset} if $|\p\A|+d_{G[T]}(\A) \leq k-1$.

We use the algorithm from \cite{N-TALG} for {\RSC}.
Two deficient bisets $\A,\B$ are {\bf $T$-dependent} if $A \cap T \subs \p\A$ or $B \cap T \subs \p\B$.
A {\RSC} instance is {\bf $T$-independence-free} if no pair of bisets are $T$-independent.
We have the following from previous work \cite{N-TALG}.

\begin{theorem} [\cite{N-TALG}] \label{t:ind}
$T$-independence-free {\RSC} instances admit ratio $O(k \ln |T|)$.
\end{theorem}

Clearly, a sufficient condition for an instance to be $T$-independence-free is:

\begin{lemma} \label{l:ind}
If for a {\RSC} instance $|A \cap T| \geq k$ holds for every deficient biset $\A$, 
then the instance is $T$-independence-free.
\end{lemma}

In the next two lemmas we show how to find an $O(k\ln n)$-approximate set $S \subs V \sem T$ such that adding $S$ to $T$ result in a
$T$-independence-free instance.

\begin{lemma} [Inflation Lemma for general graphs] \label{l:RS}
There exists a polynomial time algorithm that given an instance of {\RSC} finds
$S \subs V \sem T$ such that $|A \cap S| \geq k-(|\p\A|+d_{G[T \cup S]}(\A))$ holds for any $(T,r)$-biset $\A$,
and $w(S)=O(\ln \De) \cdot {\opt}$.
\end{lemma}
\begin{proof}
The {\sc Centered} {\RSC} problem is a particular case of the {\RSC} problem when all nodes of positive weight are neighbors of the root.
This problem admits ratio $O(\ln \De)$ \cite{N-SL}, where here $\De$ is the maximum degree of a neighbor of the root. 
We use this in our algorithm as follows:

\medskip \medskip

\begin{algorithm}[H]
\caption{$(G=(V,E),w,r,T,k)$}  
\label{alg:RS}
construct a {\sc Centered} {\RSC} instance $(G'=(V,E'),w,T,r,k)$, where 
$G'$ is obtained from $G$ by removing \ \ edges in $G[(V \sem T) \cup \{r\}]$ 
and adding an $rv$-edge for each $v \in V \sem T$ \\
compute an $O(\ln \De)$-approximate solution $S \subs V \sem T$ for the obtained {\sc Centered} {\RSC} instance \\
\Return{$S$}
\end{algorithm}

\medskip \medskip

Let $S^*$ and $S^*_c$ be optimal solutions to {\RSC} and the constructed {\sc Centered} {\RSC} instances, respectively. 
For every $t \in T$ fix some set of $k$ internally disjoint $rt$-paths in the graph $G[T \cup S^*]$, 
and obtain a set $P_t$ by picking for each path the node in $S^*$ that is closest to $t$ on this path, if such a node exists.
Let $P=\cup_{t \in T}P_t$. 
Then $P$ is a feasible solution to the constructed {\sc Centered} {\RSC} instance,
since for each $t \in T$, $G'$ has $|P_t|$  internally disjoint $rt$-paths of length $2$ each that go through $P_t$,
and $k-|P_t|$ paths that have all nodes in $T$. 
Furthermore, since $P \subs S^*$, $w(P) \leq w(S^*)$.
Thus $w(S^*_c) \leq w(P) \leq w(S^*)$, implying that $w(S) =O(\ln \De) \cdot w(S^*)$.

Now let $\A$ be a $(T,r)$-biset on $T \cup S$. 
Then:
\begin{itemize}
\item
$d_{G'[T \cup S]}(\A)=|A \cap S|+d_{G[T \cup S]}(\A)$ by the construction. 
\item
$|\p\A|+d_{G'[T \cup S]}(\A) \geq k$ since $G'[T \cup S]$ is $k$-$(T,r)$-connected.
\end{itemize}
Combining we get that $|\p\A|+d_{G[T \cup S]}(\A)+|A \cap S| \geq k$, as claimed.
\qed
\end{proof}

Our algorithms use the following simple procedure -- Algorithm~\ref{alg:p}, 
that sequentially adds $p$ sets $S_1, \ldots,S_p$ to an $m$-dominating set $T=T_0$ with $m \geq k$; 
in the case of general graphs considered in this section, each $S_i$ is as in Lemma~\ref{l:RS}.

\medskip \medskip

\begin{algorithm}[H]
\caption{$(G=(V,E),c,r,T=T_0,k,1 \leq p \leq k-1)$}  
\label{alg:p}
\For{$i=1$ to $p$}
{
$T \gets T \cup S_i$}
\Return{$T$}
\end{algorithm}


\begin{lemma} \label{l:RS'}
Suppose that we are given a {\RSC} instance such that $T$ is an $m$-dominating set in $G$ with $m \geq k$.
If at each iteration $i$ at step~2 of Algorithm~\ref{alg:p} we add to $T$ a set $S=S_i$ is as in Lemma~\ref{l:RS},
then at the end of the algorithm $w(T \sem T_0) = O(p\ln \De) \cdot {\opt}$, and
$|A \cap T| \geq m-\ell+p(k-\ell)$ holds for any biset $\A$ on $T$ with $|\p\A|+d_{G[T]}(\A)=\ell \leq k-1$.
In particular, if $p \geq \max\{2k-m-1,1\}$ then the resulting instance is $T$-independence-free.
\end{lemma}
\begin{proof}
The bound $w(T \sem T_0) = O(p\ln \De) \cdot {\opt}$ follows from Lemma~\ref{l:kc} and the 
bound $w(S)=O(\ln \De) \cdot {\opt}$ in Lemma~\ref{l:RS}.

Let $\A$ be a biset as in the lemma.
Let $T_i=T_0 \cup S_1 \cup \cdots \cup S_i$ be the set stored in $T$ at the end of iteration $i$, where $T_0$ is the initial set.
Applying Lemma~\ref{l:RS} on $T_{i-1}$ and $S_i$ we get  
$$
|A \cap S_i| \geq k-(|\p\A \cap T_{i-1}|+d_{G[T_i]}(\A)) \geq k-(|\p\A|+d_{G[T]}(\A))=k-\ell \ .
$$
In particular $A \cap S_1 \neq \empt$. Any $v \in A \cap S_1$ has in $G[T]$ at least $m$ neighbors in $T_0$,
and at most $\ell$ of them are not in $A$;
thus $v$ has at least $m-\ell$ neighbors in $A \cap T_0$, so $|A \cap T_0| \geq m-\ell$.
Since $T_0,S_1, \ldots,S_p$ are pairwise disjoint we get
$|A \cap T| \geq |A \cap T_0|+\sum_{i=1}^p|A \cap S_i| \geq m-\ell+p(k-\ell)$.
\qed
\end{proof}


The proof of the following known statement can be found in \cite{KN1}, and the second part follows from 
Mader's Undirected Critical Cycle Theorem \cite{mad-cycle}.

\begin{lemma} \label{l:H}
Let $H_r=(U,F)$ be a $k$-$(U,r)$-connected graph and $R$ the set of neighbors of $r$ in $H_r$.
The graph $H=H_r \setminus \{r\}$ can be made $k$-connected by adding a set $J$ of new edges on $R$,
and if $J$ is inclusion minimal then $J$ is~a~forest.
\end{lemma}

Note that an inclusion minimal edge set $J$ as in Lemma~\ref{l:H}
can be computed in polynomial time, by starting with $J$ being a clique on $R$ and repeatedly removing 
from $J$ an edge $e$ if $H \cup (J \setminus e)$ remains $k$-connected.

Our algorithm for general graphs is as follows.

\medskip \medskip

\begin{algorithm}[H]
\caption{$(G=(V,E),w,T)$ {\bf general graphs}}  
\label{alg:gen}
construct a graph $G_r$ by adding to $G$ and to $T$ a new node $r$ connected to a set $R \subseteq T$ of $k$ nodes
by a set $F_r=\{rv:v \in R\}$ of new edges \\
apply the Lemma~\ref{l:RS'} algorithm with $p=\max\{2k-m-1,0\}$ \\ 
use the algorithm from Theorem~\ref{t:ind} to compute an $O(k\ln n)$-approximate set $S \subs V \sem T$ 
such that $H_r=G_r[T \cup S]$ is $k$-$(T,r)$-connected \\
let $H=H \setminus \{r\}=G[T \cup S]$ and let $J$ be a forest of new edges on $R$ as in Lemma~\ref{l:H}
such that the graph $H \cup J$ is $k$-connected \\
for every $uv \in J$ find a minimum weight node set $P_{uv}$ such that $G[T \cup S \cup P_{uv}]$ 
has $k$ internally disjoint $uv$-paths; let $\displaystyle P= \bigcup_{uv \in J} P_{uv}$ \\
return $T \cup S \cup P$
\end{algorithm}

\medskip \medskip

Except step~2, the algorithm is identical to the algorithm of \cite{N-CDS} -- 
the only difference is that step~2 improves the factor invoked by step~3. 
In \cite{N-CDS} it is also proved that at the end of the algorithm $T \cup S \cup P$ is a $k$-connected set.
The dominating terms in the ratio are invoked by steps 2 and 3, and they are both $O(k \ln n)$,
while step~5 invokes just ratio $k-1$; thus the overall ratio is $O(k \ln n)$.

This concludes the proof of Theorem~\ref{t:SC} for general graphs.

\section{Unit disc graphs} \label{s:udg}

Our goal in this section is to prove the following:

\begin{lemma} \label{l:udg}
Consider a {\SC} instance on unit disc graph $G=(V,E)$ where $T$ is an $m$-dominating set with $m \geq k$
and $G[T]$ is $\ell$-connected, $\ell \leq k-1$.
Then for any integer $1 \leq p \leq \ell+1$ there exists a polynomial time algorithm that computes $S \subs V \sem T$ 
such that $G[T \cup S]$ is $(\ell+1)$-connected and 
$$
\f{w(S)}{\opt}=\f{O(\ln k)}{k-\ell} \left(p+\f{(m+p)^2}{(m+p-\ell)^2}\right) \ .
$$ 
Furthermore, in the case of unit weights 
$\f{w(S)}{\opt} = \f{O(\ln k)}{k-\ell}\left(p+\f{m+p}{m+p-\ell}\right)$.
\end{lemma}

Let us show that Lemma~\ref{l:udg} implies the unit disc part of Theorem~\ref{t:SC}.
We can apply the Lemma~\ref{l:udg} algorithm sequentially, starting with an $O(1)$-approximate 
$m$-dominating set $T=T_0$, and at iteration $\ell=0,\ldots,k-1$ add to $T$ a set $S=S_\ell$ as in the lemma.
In the case of arbitrary weights choosing $p=1$ if  $m-\ell \geq \ell^{2/3}$ and $p=\ell^{2/3}$ otherwise
gives $\f{w(S_\ell)}{\opt} = \f{O(\ln \ell)}{k-\ell}\min\left\{\f{m}{m-\ell},\ell^{2/3}\right\}$. Then
denoting $S=S_0 \cup \cdots \cup S_{k-1}$ we get:
$$
\f{w(S)}{\opt} = \sum_{\ell=0}^{k-1}\f{O(\ln \ell)}{k-\ell}\min\left\{\f{m}{m-\ell},\ell^{2/3}\right\} =
O(\ln^2 k) \cdot \min\left\{\f{m}{m-k+1},k^{2/3}\right\}
$$
In the case of unit weights, choosing $p=1$ if  $m-\ell \geq \sqrt{\ell}$ and $p=\sqrt{\ell}$ otherwise gives 
$\f{w(S_\ell)}{\opt}=\f{O(\ln \ell)}{k-\ell}\min\left\{\f{m}{m-\ell},\sqrt{\ell}\right\}$,
and then by a similar analysis we get $\f{w(S)}{\opt}=O(\ln^2 k) \cdot \min\left\{\f{m}{m-k+1},\sqrt{k}\right\}$. \\

In the rest of this section we prove Lemma~\ref{l:udg}, so let $G$, $T$, and $\ell$ be as in the lemma.
We need some definitions and known facts concerning biset families.

\begin{definition}
A biset $\A$ on $V$ is {\bf deficient} (w.r.t. $T$) if 
$A \cap T \neq \empt \neq T \sem A^+$, $d_{G[T]}(\A)=0$, and $|\p\A|=\ell$.
Let $\DD_T$ denote the set of deficient bisets.
We say that $v \in V \sem T$ covers $\A \in \DD_T$ if $v \in \Ga(A) \sem T$, where $\Ga(A)$ denotes the set of neighbors of $A$ in $G$;
$S \subs V \sem T$ covers $\FF \subs \DD_T$ if every $\A \in \FF$ is covered by some $v \in S$.
\end{definition}

Note that $\A \in \DD_T$ if and only if $\p\A \subs T$ is a minimum node cut of $G[T]$, and $A \cap T$ 
is a union of some, but not all, connected components of $G[T] \sem \p\A$.
The following lemma is known, c.f. \cite{N-TALG,Fu}.

\begin{lemma} 
$G[T \cup S]$ is $(\ell+1)$-connected if and only if $S$ covers $\DD_T$.
\end{lemma}

Thus we have the following LP-relaxation 
for the problem of finding a min-weight cover of $\DD_T$:

\[\begin{array} {lllll} 
& \tau(\DD_T) = & \min              & \sum_{v \in V \sem T} w_v x_v                                  & \\
&                         & \ \mbox{s.t.} & \sum_{v \in \Gamma(A) \sem T} x_v \geq 1 \ \ \ \ \ & \forall A \in \DD_T  \\
&                         &                       & x_v \geq 0                                                                    & \forall v \in V \sem T
\end{array}\]

Note that if $\A \in \DD_T$ then $\Gamma(A) \sem T=\Gamma(A) \sem \p\A$, 
and thus the constraint $\sum_{v \in \Gamma(A) \sem T} x_v \geq 1$ is equivalent to 
$\sum_{v \in \Gamma(A) \sem \p\A} x_v \geq 1$. 

\begin{definition} \label{d:bisets}
We say that {\bf $\A$ contains $\B$} and write $\A \subs \B$ if $A \subs B$ and $A^+ \subs B^+$.
Inclusion minimal members of a biset family $\FF$ are called {\bf $\FF$-cores}.
\end{definition}

\begin{theorem} [Zhang, Zhou, Mo, \& Du \cite{ZW}] \label{t:ZW} \ 
Any $k$-connected unit disc graph has a $k$-connected spanning subgraph of maximum degree $\leq 5k$.
\end{theorem}

\begin{lemma} [Inflation Lemma for unit disc graphs] \label{l:SC}
There exists a polynomial time algorithm that 
computes $S \subs V \sem T$ that covers the family $\CC$ of cores of $\DD_T$ and $w(S)=O(\ln k) \cdot \f{{\opt}}{k-\ell}$.
\end{lemma}
\begin{proof}
The problem of covering $\CC$ is essentially a (weighted) {\SeC} problem  
where for each $v \in V \sem T$ the corresponding set has weight $w_v$ and 
consists of the cores covered by $v$. Then the greedy algorithm for {\SeC} computes a solution 
of weight $O(\ln |\CC|)$ times the value of the standard {\SeC} LP
\[ \displaystyle
\begin{array} {lllll} 
& \tau(\CC)=& \min              & \sum_{v \in V \sem T} w_v x_v                                   & \\
&                  & \ \mbox{s.t.} & \sum_{v \in \Gamma(A) \sem T} x_v \geq 1 \ \ \ \ \ & \forall A \in \CC  \\
&                  &                       & x_v \geq 0                                                                    & \forall v \in V \sem T
\end{array}
\]
For any $S' \subs V \sem T$ such that $G[T \cup S']$ is $k$-connected, any tight set $A$ has at least $k-\ell$ neighbors in $G[T \cup S']$,
hence if $x'$ is a characteristic vector of $S'$ then $\f{x'}{k-\ell}$ is a feasible solution to the LP. 
Consequently, $\tau(\CC) \leq \f{\opt}{k-\ell}$.

In [\cite{jor}, Lemma 3.5, Case II] it is proved that if $C \cap C'$ for two distinct cores $\C,\C'$ then there is $P \subs V$ with 
$|P| \leq \ell+1$ such that $P \cap C$ for every $\C \in \CC$. 
In this case $G$ has at most $\ell(\ell+1)$ distinct cores, since for every core $\C$ there is $s \in P \cap C$ and $t \in V \sem C^+$,
and for each $(s,t) \in P \times P$ there is at most one such core.
Hence in the case $|\CC| \leq \ell(\ell+1)$ we get a solution of weight $O(\ln \ell) \cdot \tau(\CC) =O(\ln \ell) \cdot \f{\opt}{k-\ell}$.

In the case $|\CC_\FF| > \ell(\ell+1)$ we have $C \cap C'=\empt$ for any $\C,\C' \in \CC$,
and relying on Theorem~\ref{t:ZW} we modify this reduction 
such that every $v \in V \sem T$ can cover at most $5k$ cores; this is essentially the {\SeC} with (soft) capacities problem.
Specifically, for each pair $(v,J)$ where $v \in V \sem S$ and $J$ is a set of at most $5k$ edges incident to $v$,
we add a new node $v_J$ of weight $w_v$ with corresponding copies of the edges in $J$.
In the obtained {\SeC} instance the maximum size of a set is at most $5k$, since the $\FF$-cores are pairwise disjoint.
Note that we do not need to construct this {\SeC} instance explicitly to run the greedy algorithm -- we just need to determine 
for each $v \in V$ the maximum number of at most $5k$ not yet covered cores that can be covered by $v$.
Since the $\FF$-cores are pairwise disjoint, this can be done in polynomial time.
Note that during the greedy algorithm we may pick pairs $(v,J)$ and $(v,J')$ with distinct $J,J'$ but with the same node $v$,
but this only makes the solution lighter. 
Since in the {\SeC} instance the maximum set size is $5k$, the computed solution has weight $O(\ln k) \cdot \tau$,
where here $\tau$ is an optimal LP-value of the modified instance.
Now we argue in the same way as before that $\tau \leq \f{{\opt}}{k-\ell}$.
Consider a feasible solution $S' \subs V \sem T$ and an edge $J'$ such that $G[T] \cup S' \cup J'$ is a spanning 
$k$-connected subgraph of $G[T \cup S']$ and $\deg_{J'}(v) \leq 5k$ for all $v \in S'$; such $J'$ exists by Theorem~\ref{t:ZW}.
Let $x'$ be the characteristic vector of the pairs $(v,J'_v)$ where 
$v \in S'$ and $J'_v$ is the set of edges in $J'$ incident to $v$. 
Any tight set $A$ has at least $k-\ell$ neighbors in $G[T] \cup S' \cup J'$,
hence $\f{x'}{k-\ell}$ is a feasible solution to the LP. 
Consequently, $\tau \leq \f{{\opt}}{k-\ell}$.
\qed
\end{proof}

\begin{corollary} \label{c:CS}
If at step~2 of Algorithm~\ref{alg:p} we add $S=S_i$ is as in Lemma~\ref{l:SC}, then at the end of the algorithm 
$w(T \sem T_0) = O(p\ln k) \cdot \tau^*$ and $|A \cap T| \geq m-\ell+p$ holds for any $\A \in \DD_T$.
\end{corollary}
\begin{proof}
We have $|A \cap S_i| \geq 1$ for all $i$.
In particular $A \cap S_1 \neq \empt$. Any $v \in A \cap S_1$ has in $G[T]$ at least $m$ neighbors in $T_0$,
and at most $\ell$ of them are not in $A$;
thus $v$ has at least $m-\ell$ neighbors in $A \cap T_0$, so $|A \cap T_0| \geq m-\ell$.
Since $T_0,S_1, \ldots,S_p$ are pairwise disjoint we get
$|A \cap T| \geq |A \cap T_0|+\sum_{i=1}^p|A \cap S_i| \geq m-\ell+p$.
\qed
\end{proof}

Now we decompose the problem of covering $\DD_T$ into several subproblems.
For $r \in T$ let $\DD_{(T,r)}=\{\A \in \DD_T:r \in T \sem A^+\}$.

\begin{theorem} [\cite{N-subs}] \label{t:subs}
Given an $\ell$-$T$-connected graph with $|T| \geq \ell+1$, one can find in polynomial time 
$R \subs T$ of size $|R|=O\left(\f{|T|}{|T|-\ell} \ln\ell\right)$ such that $\DD = \cup_{r \in R} \DD_{(T,r)}$.
\end{theorem}

We now describe how to cover the family $\DD_{(T,r)}$ for given $r \in T$.

\begin{definition}
The {\bf intersection} and the {\bf union} of two bisets $\A,\B$ are defined by
$\A \cap \B = (A \cap B, A^+ \cap B^+)$ and $\A \cup \B=(A \cup B,A^+ \cup B^+)$.
The biset $\A \sem \B$ is defined by $\A \sem \B=(A \sem B^+,A^+ \sem B)$.
A biset family $\FF$ is called:
\begin{itemize}
\item
{\bf uncrossable} if $\A \cap \B,\A \cup \B \in \FF$ or if $\A \sem \B,\B \sem \A \in \FF$ for all $\A,\B \in \FF$.
\item 
{\bf $T$-intersecting} if $\A \cap \B,\A \cup \B \in \FF$ for any $\A,\B \in \FF$ with $A \cap B \cap T \neq \empt$.
\item
{\bf $T$-co-crossing} if $\A \sem \B,\B \sem \A \in \FF$ for any $\A,\B \in \FF$ with $A \cap B^* \cap T \neq \empt$ and $B \cap A^* \cap T \neq \empt$.  
\end{itemize}
\end{definition}

\begin{lemma} [\cite{N-TALG}] \label{l:tic}
$\DD_{(T,r)}$ is $T$-intersecting and $T$-co-crossing for any $r \in T$.
\end{lemma}

\begin{theorem} [\cite{N-TALG}] \label{t:decomp}
There exists a polynomial time algorithm that given a $T$-in\-ter\-secting $T$-co-crossing biset family $\FF$ sequentially
finds $O\left(\f{q+\ell}{q}\right)$ $T$-intersecting uncrossable subfamilies of $\FF$,
such that the union of covers of these subfamilies covers $\FF$, 
where $q=\min\{|A \cap T|:\A \in \FF\}$ and $\displaystyle \ell=\max_{\A \in \FF} |\p\A \cap T|$.
\end{theorem}

\begin{theorem} [Fukunaga \cite{Fu}] \label{t:Fu} \ 
If $\FF$ is a $T$-intersecting uncrossable subfamily of $\DD_T$ 
then there exists a polynomial time algorithm that computes a cover $S$ of $\FF$ of weight $w(S) \leq \f{15 \opt}{k-\ell}$.
\end{theorem}

Combining Lemma~\ref{l:tic} with Theorems \ref{t:decomp} and \ref{t:Fu} we get:

\begin{corollary} \label{c:root}
For any $r \in T$, there exists a polynomial time algorithm that computes a cover $S_r$ of $\DD_{(T,r)}$ such that
if $q=\min\{|A \cap T|:\A \in \DD_{(T,r)}\}$ then 
$$
\f{w(S_r)}{\opt} = O\left(\f{q+\ell}{q(k-\ell)}\right)  \ .
$$
\end{corollary}

The algorithm for unit disc graphs is as follows:

\medskip \medskip

\begin{algorithm}[H]
\caption{$(G=(V,E),w,T,p)$ {\bf unit disc graphs}}  
\label{alg:udg}
apply Algorithm~\ref{alg:p} where at step~2 each $S_i$ is as in Lemma~\ref{l:SC} \\
{\bf if} $G[T]$ is $(\ell+1)$ connected then $S \gets \empt$  \\
\Else( \ \ \ \{now $|T| \geq m+p$ and $|A \cap T| \geq m+p-\ell$ for all $\A \in \DD_T$\}) 
{
find a set $R$ of $O\left(\f{|T|}{|T|-\ell} \ln\ell\right)$ roots as in Theorem~\ref{t:subs} \\
for each $r \in R$ compute a cover $S_r$ of $\DD_{(T,r)}$ as in Corollary~\ref{c:root} \\
$S \gets \cup_{r \in R}S_r$
}
\Return{$T \cup S$}
\end{algorithm}

\medskip \medskip

We bound the weight of each of the sets computed. 
Let $T_0$ denote the initial set stored in $T$. By Lemma~\ref{l:SC}, at the end of step~1 we have 
$$
\f{w(T \sem T_0)}{\opt} =\f{O(\ln k)}{k-\ell} \cdot p
$$
Now we bound the weight of the set $S$ computed in steps $3$ to $6$:
$$
\f{w(S)}{\opt}= |R| \cdot  O\left(\f{q+\ell}{q(k-\ell)}\right)=\f{O(\ln \ell)}{k-\ell}\f{|T|}{|T|-\ell}\f{q+\ell}{q}=
\f{O(\ln k)}{k-\ell}\f{(m+p)^2}{(m+p-\ell)^2}
$$
The overall weight of the augmenting set computed is as claimed in Lemma~\ref{l:udg}.

In the case of unit weights, we add arbitrary $\ell$ nodes to $T$; this step invokes an additive term of $O(1)$ to the ratio.
Then we will have $|R|=O(\ln \ell)$ and thus
$$
\f{w(S)}{\opt}= |R| \cdot  O\left(\f{q+\ell}{q(k-\ell)}\right)=\f{O(\ln \ell)}{k-\ell}\f{q+\ell}{q}=
\f{O(\ln k)}{k-\ell}\f{m+p}{m+p-\ell}
$$
The overall weight of the augmenting set computed is as claimed in Lemma~\ref{l:udg}.

This concludes the proof of Lemma~\ref{l:udg}, and thus also the proof of Theorem~\ref{t:SC}.


\end{document}